\documentclass[letter,11pt]{article}
\usepackage[ansinew]{inputenc}
\usepackage{amssymb,amsmath,amsthm}
\usepackage{graphicx}
\usepackage{url}
\usepackage[hmargin=0.8in]{geometry}
\usepackage{subfigure}
 \usepackage{tabularx}
   
\usepackage{enumitem}
\usepackage[numbers,sort&compress,sectionbib]{natbib}

\usepackage{hyperref}

\newtheorem{theorem}{Theorem}

\newtheorem{lemma}[theorem]{Lemma}

\newcommand{\krulehead}{{\bf Rule \arabic{krule}}}

\title{Max-Cut Parameterized Above the Edwards-Erd\H{o}s Bound}

\author{
 Robert Crowston\thanks{%
  Royal Holloway - University of London, United Kingdom.
  \texttt{robert@cs.rhul.ac.uk}}
 \and Mark Jones\thanks{%
  Royal Holloway - University of London, United Kingdom.
  \texttt{markj@cs.rhul.ac.uk}}
 \and Matthias Mnich\thanks{%
  Cluster of Excellence, Saarbr\"{u}cken, Germany.
  \texttt{m.mnich@mmci.uni-saarland.de}}
}

\date{}

\begin{document}

\maketitle

\begin{abstract}
  We study the boundary of tractability for the {\sc Max-Cut} problem in graphs.
  Our main result shows that {\sc Max-Cut} parameterized above the Edwards-Erd\H{o}s bound is fixed-parameter tractable:
  we give an algorithm that for any connected graph with $n$ vertices and $m$ edges finds a cut of size
  \begin{equation*}
    \frac{m}{2} + \frac{n-1}{4} + k
  \end{equation*}
  in time $2^{O(k)}\cdot n^4$, or decides that no such cut exists.
  
  This answers a long-standing open question from parameterized complexity that has been posed a number of times over the past 15 years.
  
  Our algorithm is asymptotically optimal, under the Exponential Time Hypothesis, and is strengthened by a polynomial-time computable kernel of polynomial size.

\end{abstract}

\section{Introduction}
\label{sec:introduction}
The study of cuts in graphs is a fundamental area in theoretical computer science, graph theory, and polyhedral combinatorics, dating back to the 1960s.
A \emph{cut} of a graph is an edge-induced bipartite subgraph, and its \emph{size} is the number of edges it contains.
Finding cuts of maximum size in a given graph was one of Karp's famous 21 $\mathsf{NP}$-complete problems~\cite{Karp1972}.
Since then, the {\sc Max-Cut} problem has received considerable attention in the areas of approximation algorithms, random graph theory, combinatorics, parameterized complexity, and others; see the survey~\cite{PoljakTuza1995}.

As a fundamental $\mathsf{NP}$-complete problem, the computational complexity of {\sc Max-Cut} has been intensively scrutinized.
We continue this line of research and explore the boundary between tractability and hardness, guided by the question: \emph{Is there a dichotomy of computational complexity of {\sc Max-Cut} with respect to the size of the maximum cut?}

This question was already studied by Erd\H{o}s~\cite{Erdos1965} in the 1960s, who gave a randomized polyno\-mial-time algorithm that in any $n$-vertex graph with $m$ edges finds a cut of size at least $m/2$.
Erd\H{o}s~\cite{Erdos1965,Erdos1967} also (erroneously) conjectured that the value $m/2$ can be raised to $m/2 + \varepsilon m$ for some~$\varepsilon > 0$; only much later it was shown~\cite{HaglinVenkatasan1991,NgocTuza1993} that finding cuts of size $m/2 + \varepsilon m$ is $\mathsf{NP}$-hard for every $\varepsilon > 0$.
Furthermore, the {\sc Max-Cut Gain} problem---maximize the gain compared to a random solution that cuts $m/2$ edges---does not allow constant approximation~\cite{KhotODonnell2009} under the Unique Games Conjecture, and the best one can hope for is to cut a $1/2 + \Omega(\varepsilon/\log(1/\varepsilon))$ fraction of edges in graphs in which the optimum is $1/2 + \varepsilon$~\cite{CharikarWirth2004}.

However, the lower bound $m/2$ \emph{can} be increased, by a sublinear function: Edwards~\cite{Edwards1973,Edwards1975} in 1973 proved that in connected graphs a cut of size
\begin{equation}
\label{eqn:eeboundmixed}
  m/2 + (n-1)/4,
\end{equation}
always exists.
Thus, any graph with $n$ vertices, $m$ edges and $t$ components has a cut of size at least $m/2 + (n-t)/4$.
The lower bound~\eqref{eqn:eeboundmixed} is famously known as the \emph{Edwards-Erd\H{o}s bound}, and it is tight for cliques of every odd order $n$.
The bound has been proved several times~\cite{PoljakTurzik1982,NgocTuza1993,ErdosEtAl1997,BollobasScott2002,CrowstonEtAl2011}, with some proofs yielding polynomial-time algorithms to attain it.

As~\eqref{eqn:eeboundmixed} is tight for infinitely many non-isomorphic graphs, and finding maximum cuts is $\mathsf{NP}$-hard, raising the lower bound~\eqref{eqn:eeboundmixed} even further requires a new approach: a \emph{fixed-parameter algorithm}, that for any connected graph with $n$ vertices and $m$ edges, and integer $k\in\mathbb N$, finds a cut of size at least $m/2 + (n-1)/4 + k$ (if such exists) in time $f(k)\cdot n^c$, where~$f$ is an arbitrary function dependent only on $k$ and $c$ is an absolute constant independent of~$k$.
The point here is to confine the combinatorial explosion to the (small) parameter $k$.
But at first sight, it seems not even clear how to find a cut of size $m/2 + (n-1)/4 + k$ in time~$n^{f(k)}$, for any function~$f$.
We refer to the books by Downey and Fellows~\cite{DowneyFellows1999} and Flum and Grohe~\cite{FlumGrohe2006} for background on parameterized complexity.

In 1997, Mahajan and Raman~\cite{MahajanRaman1997} gave a fixed-parameter algorithm for the variant of this problem with Erd\H{o}s' lower bound $m/2$, and showed how to decide existence of a cut of size $m/2 + k$ in time $2^{O(k)}\cdot n^{O(1)}$.
Their result was strengthened by Bollob\'{a}s and Scott~\cite{BollobasScott2002} who replaced $ m/2$ by the stronger bound
\begin{equation}
\label{eqn:eeboundedges}
  m/2 + \frac{1}{8}\left(\sqrt{8m + 1} - 1\right).
\end{equation}
For unweighted graphs, this bound is weaker than \eqref{eqn:eeboundmixed}.
It remained an open question (\cite{MahajanRaman1997,MahajanEtAl2009,Sikdar2010,GutinYeo2010,CrowstonEtAl2011}) whether this result could be strengthened further by replacing \eqref{eqn:eeboundedges} with the stronger bound~\eqref{eqn:eeboundmixed}.

\subsection{Main Results}
\label{sec:mainresults}

We settle the computational complexity of {\sc Max-Cut} above the Edwards-Erd\H{o}s bound~\eqref{eqn:eeboundmixed}.
\begin{theorem}
\label{thm:maxcutaboveeeisfpt}
  There is an algorithm that computes, for any connected graph $G$ with $n$ vertices and~$m$ edges and any integer $k\in \mathbb N$, in time $2^{O(k)}\cdot n^4$ a cut of $G$ with size at least $m/2 + (n-1)/4 + k$, or decides that no such cut exists.
\end{theorem}
Theorem~\ref{thm:maxcutaboveeeisfpt} answers a question posed several times over the past 15 years~\cite{MahajanRaman1997,MahajanEtAl2009,Sikdar2010,GutinYeo2010,CrowstonEtAl2011}.
In particular, instances with~$k = O(\log m)$ can be solved efficiently, which thus extends the classical polynomial-time algorithms~\cite{NgocTuza1993,ErdosEtAl1997,BollobasScott2002,CrowstonEtAl2011} that compute a cut of size~\eqref{eqn:eeboundmixed}.

The running time of our algorithm is likely to be optimal, as the following theorem shows.
\begin{theorem}
\label{thm:ethlowerbound}
  No algorithm can find cuts of size $m/2 + (n-1)/4 + k$ in time $2^{o(k)}\cdot n^{O(1)}$ given a connected graph with $n$ vertices and $m$ edges, and integer $k\in\mathbb N$, unless the Exponential Time Hypothesis fails.
\end{theorem}
The Exponential Time Hypothesis was introduced by Impagliazzo and Paturi~\cite{ImpagliazzoPaturi1999}, and states that $n$-variable $3$-CNF-SAT formulas cannot be solved in time subexponential in $n$.

Fixed-parameter tractability of {\sc Max-Cut} above Edwards-Erd\H{o}s bound~$m/2 + (n-1)/4$ implies the existence of a so-called \emph{kernelization}, which is an algorithm that efficiently compresses any instance $(G,k)$ into an equivalent instance $(G',k')$, the \emph{kernel}, whose size $g(k) = |G'| + k'$ itself depends on $k$ only.
Alas, the size $g(k)$ of the kernel for many fixed-parameter tractable problems is enormous, and in particular many fixed-parameter tractable problems do not admit kernels of size \emph{polynomial} in~$k$ unless $\mathsf{coNP}\subseteq \mathsf{NP}/\textnormal{poly}$~\cite{BodlaenderEtAl2009}.
We prove the following.
\begin{theorem}
\label{thm:polykernelstrongerbound}
  There is a polynomial-time algorithm that, for any integer $k\in\mathbb N$, compresses any connected graph $G = (V,E)$ with integer $k\in\mathbb N$ to a connected graph $G' = (V',E')$  of order $O(k^5)$, and produces an integer  $k'\leq k$, such that $G$ has a cut of size $|E|/2 + (|V|-1)/4 + k$ if and only if $G'$ has a cut of size $|E'|/2 + (|V'|-1)/4 + k'$, for some $k'\leq k$.
\end{theorem}


%
%
%
Note that Theorems~\ref{thm:maxcutaboveeeisfpt} and~\ref{thm:polykernelstrongerbound} only hold for unweighted graphs; determining the parameterized complexity of the weighted versions remains open.

\subsection{Our Techniques and Related Work}
A number of standard approaches that have been developed for ``above-guarantee'' parameterizations of other problems are unavailable for this problem; hence, our algorithm differs significantly from others in the area.
The most common approach is to use probabilistic analysis of a random variable whose expected value corresponds to a solution matching the guarantee.
However, there is no simple randomized procedure known giving a cut of size $m/2 + (n-1)/4$.
Another approach is to make use of approximation algorithms that give a factor-$c$ approximation, when the problem is parameterized above the bound $c \cdot n$; here, $n$ is the maximum value of the objective
function; yet there is no such approximation algorithm for the {\sc Max-Cut} problem. 


Instead, we make use of ``one-way'' reduction rules.
Unlike standard reduction rules,  our rules do not produce an equivalent instance, but merely preserve ``no''-instances.
If they produce a ``yes''-instance we know the original instance is also a ``yes''-instance, but if they do not, then this fact allows us to show the original instance has a relatively simple structure, which we can then use to solve the problem. 
We have not seen this approach used in parameterized algorithmics before, and it may prove useful in solving other problems where useful two-way reduction rules are hard to find.

Our results are based on algorithmic as well as combinatorial arguments.
To prove Theorem~\ref{thm:maxcutaboveeeisfpt}, we use one-way reduction rules to reduce the problem to one on graphs which are close to being in a special class of graphs we call ``clique-forests'', on which we show how to solve the problem efficiently.
Theorem~\ref{thm:ethlowerbound} follows from a straightforward parameterized reduction.
Theorem~\ref{thm:polykernelstrongerbound} is proven by a careful analysis of random cuts via the probabilistic method.

%

For variants of {\sc Max-Cut}, the ``boundary of tractability'' above guaranteed values has also been investigated in the setting of parameterized complexity.
For instance, in {\sc Max-Bisection} we seek a cut such that the number of vertices in both sides of the bipartition is as equal as possible, the tight lower bound on the bisection size is only~$m/2$; fixed-parameter tractability of {\sc Max-Bisection} above $m/2$ was recently shown by Gutin and Yeo~\cite{GutinYeo2010} and Mnich and Zenklusen~\cite{MnichZenklusen2012}.

\section{Preliminaries}
\label{sec:preliminaries}
With the exception of the term ``clique-forest'' (see below), we use standard graph theory terminology and notation.
Given a graph $G$, let $V(G)$ be the vertices of $G$ and let $E(G)$ be the edges of $G$.
For disjoint sets $S,T\subseteq V(G)$, let $E(S,T)$ denote the set of edges in $G$ with one vertex in $S$ and one vertex in $T$.
For $S\subseteq V(G)$, let $G[S]$ denote the subgraph induced by the vertices of $S$, and let $G - S$ denote the graph $G[V(G)\setminus S]$.
We say that $G$ has a \emph{cut of size $t$} if there exists a set $S \subseteq V(G)$ such that $|E(S, V(G) \setminus  S)|=t$.
The graph $G$ is \emph{connected} if any two of its vertices are connected by a path, and it is \emph{2-connected} if $G - \{v\}$ is connected for every $v\in V(G)$.
A vertex $v$ is a \emph{cut-vertex} for $G$ if $G$ is connected and $G - \{v\}$ is disconnected.

We study the following formulation of {\sc Max-Cut} parameterized above Edwards-Erd\H{o}s bound:

\begin{center}
\fbox{~\begin{minipage}{0.9\textwidth}
\textsc{Max-Cut} above Edwards-Erd\H{o}s ({\sc Max-Cut-AEE})

\emph{Instance}: A connected graph $G$ with $n = |V(G)|$ vertices and $m = |E(G)|$ edges, and an integer $k\in\mathbb N$.

\smallskip

\emph{Parameter}: $k$.

\smallskip

\emph{Question}:
Does $G$ have a cut of size at least $\frac{m}{2}+\frac{n-1}{4} + \frac{k}{4}$?
\end{minipage}~}
\end{center}

We ask for a cut of size $\frac{m}{2}+\frac{n-1}{4}+\frac{k}{4}$, rather than the more usual $\frac{m}{2}+\frac{n-1}{4} + k$, so that we may treat~$k$ as an integer at all times.
Note that this does not affect the existence of a fixed-parameter algorithm or polynomial-size kernel.

An \emph{assignment} or \emph{coloring} on $G$ is a function $\alpha: V(G) \rightarrow \{\mathsf{red}, \mathsf{blue}\}$, and an edge is \emph{cut} or \emph{satisfied} by $\alpha$ if one of its vertices is colored $\mathsf{red}$ and the other vertex is colored $\mathsf{blue}$.
Note that a graph has a cut of size $t$ if and only if it has an assignment that satisfies at least $t$ edges.
A \emph{partial assignment} on $G$ is a function~$\alpha: X \rightarrow \{\mathsf{red}, \mathsf{blue}\}$, where $X$ is a subset of $V(G)$.

A \emph{block} in $G$ is a maximal $2$-connected subgraph of $G$. It is well known that the blocks of a graph can be found in linear time, and that for a connected graph, if two blocks intersect then their intersection is a cut-vertex. It follows that any cycle is contained in a single block. We say a block is a \emph{leaf-block} if it has at most one vertex which is contained in another block. Note that every non-empty graph has a leaf-block.

We define a class of graphs, \emph{clique-forests}, as follows.
A clique is a clique-forest, as is an empty graph.
The disjoint union of any clique-forests is a clique-forest.
Finally, a graph formed from a clique forest by identifying two vertices from separate components is also a clique-forest.
Note that the clique-forests are exactly the graphs for which every block is a clique.
Such graphs are sometimes called \emph{block graphs}, but this term has contradictory definitions in the literature\footnote{e.g. the definition in by Bandelt and Mulder~\cite{Bandelt1986182} agrees with our definition of clique-forest, but the definition in Chapter $3$ of Diestel's book~\cite{diestel2005graph} has that every block graph is a tree.} and we will not use it in this paper.

%

 To arrive at a realistic analysis of the required computational effort, throughout our model of computation is the random access machine with the restriction that arithmetic operations are considered unit-time only for constant-size integers; in this model, two $b$-bit integers can be added, subtracted, and compared in $O(b)$ time.

\section{Fixed-Parameter Algorithm for Max-Cut above the \texorpdfstring{Edwards-Erd\H{o}s}{Edwards-Erdos} Bound}
\label{sec:fptalgorithm}
%
%

In this section, we prove Theorem~\ref{thm:maxcutaboveeeisfpt}. Central to this proof is the following lemma, which also forms the basis of our kernel
in Theorem~\ref{thm:polykernelstrongerbound}.

\begin{lemma}
  \label{thm:findS}
  Given a connected graph $G$ with $n$ vertices and $m$ edges and an integer $k$, we can in polynomial time decide that either $G$ has a cut of size at least $\frac{m}{2} + \frac{n-1}{4} + \frac{k}{4}$, or find a set $S$ of at most $3k$ vertices in $G$ such that 
$G  -  S$ is a clique-forest.
\end{lemma}

Given a set $S$ as described in Lemma \ref{thm:findS}, after guessing a coloring of $S$ we reduce \textsc{Max-Cut-AEE} to a related problem on clique-forests, which we show is polynomial time solvable using Lemma \ref{thm:poly}. As there are at most $2^{3k}$ possible colorings of $S$, we get an algorithm for \textsc{Max-Cut-AEE} with the required running time.

To prove Lemma \ref{thm:findS}, we use a set of ``one-way'' reduction rules. These rules produce an instance $(G',k')$ such that if $(G',k')$ is a ``yes''-instance then $(G,k)$ is also a ``yes''-instance; this is shown in Lemma \ref{thm:validity}. The converse does not necessarily hold; if $(G',k')$ is a ``no''-instance then $(G,k)$ may be a ``yes''- or ``no''-instance. We show in Lemma \ref{thm:exhaustion} that $(G',k')$ will be an instance that is trivially easy to solve (one with no edges), so if $(G',k')$ is a ``yes''-instance then we are done. Otherwise, the reduction rules mark a set of vertices in $G$, and we will show using Lemma \ref{thm:treeOfCliques} that if $S$ is the set of marked vertices, then $|S|\le 3k$ and $G  -  S$ is a clique-forest, as required.

\smallskip
\refstepcounter{krule}
\label{rule:vertexclique}
\noindent
\framebox[\textwidth][l]{
  \begin{tabularx}{0.98\textwidth}{lX}
    \krulehead:     & Apply to a connected graph $G$ with  $v \in V(G), X \subseteq V(G)$ such that $X$ is a connected component of $G - \{v\}$ and $X \cup \{v\}$ is a clique.\\
    Remove:        & All vertices in $X$.\\
    Mark: & Nothing.\\
    Parameter:     & Reduce $k$ by $1$ if $|X|$ is odd, otherwise leave $k$ the same.\\
  \end{tabularx}
}

\smallskip
\refstepcounter{krule}
\label{rule:vertexmark}
\noindent
\framebox[\textwidth][l]{
  \begin{tabularx}{0.98\textwidth}{lX}
    \krulehead:     & Apply to a connected graph $G$ reduced by Rule \ref{rule:vertexclique} with $v \in V(G),X \subseteq V(G)$ such that $X$ is a connected component of $G - \{v\}$ and $X$ is a clique.\\

    Remove:        & All vertices in $X$.\\
    Mark: & $v$.\\
    Parameter:     & Reduce $k$ by $2$.\\
  \end{tabularx}
}

\smallskip
\refstepcounter{krule}
\label{rule:triplet}
\noindent
\framebox[\textwidth][l]{
  \begin{tabularx}{0.98\textwidth}{lX}
    \krulehead:     & Apply to a connected graph $G$ with $a,b,c \in V(G)$ such that $\{a,b\}, \{b,c\} \in E(G), \{a,c\} \notin E(G)$, and $G  -  \{a,b,c\}$ is connected.\\
    Remove:        & The vertices $a,b,c$.\\
    Mark: & $a,b,c$.\\
    Parameter:     & Reduce $k$ by $1$.\\
  \end{tabularx}
}

\smallskip
\refstepcounter{krule}
\label{rule:pincer}
\noindent
\framebox[\textwidth][l]{
  \begin{tabularx}{0.98\textwidth}{lX}
    \krulehead:     & Apply to a connected graph $G$ with $x,y \in V(G)$ such that $\{x,y\} \notin E(G)$, $G  - \{x,y\}$ has two connected components, $X$ and $Y$, and $X \cup \{x\}$ and $X \cup \{y\}$ are cliques.\\
    Remove:        & Vertices $\{x,y\}\cup X$.\\
    Mark: & $x, y$.\\
    Parameter:     & Reduce $k$ by $1$.\\
\end{tabularx}
}

\smallskip
These rules can be applied exhaustively in polynomial time, as each rule reduces the number of vertices in $G$, and for each rule we can check for any applications of that rule by trying every set of at most three vertices in $V(G)$ and examining the components of the graph when those vertices are removed.

\begin{lemma} \label{thm:validity}
  Let $(G,k)$ and $(G',k')$ be instances of \textsc{Max-Cut-AEE} such that $(G',k')$ is reduced from $(G,k)$ by an application of Rules \ref{rule:vertexclique}, \ref{rule:vertexmark}, \ref{rule:triplet} and \ref{rule:pincer}.
  Then $G'$ is connected, and if $(G',k')$ is a ``yes''-instance of {\sc Max-Cut-AEE} then so is $(G,k)$.
\end{lemma}
\begin{proof}
  First, we show that $G'$ is connected.
  For Rules \ref{rule:vertexclique} and \ref{rule:vertexmark}, observe that for $s,t \in V(G)  \setminus  X$, no path between $s$ and $t$ passes through $X$, so $G  -  X$ is connected.
  For Rule \ref{rule:triplet}, the conditions explicitly state that we only apply the rule if the resulting graph is connected.
  For Rule \ref{rule:pincer}, observe that we remove all vertices except those in $Y$, and $Y$ forms a connected component.

  Second, we prove separately for each rule the following claim, in which $n'$ denotes the number of vertices and $m'$ the number of edges removed by the rule.
  \begin{equation}
  \label{eqn:starclaim}
    \begin{array}{l}
     \mbox{Any assignment to the vertices of $G'$ can be extended to an }\\
     \mbox{assignment on $G$ that cuts an additional $\frac{m'}{2} + \frac{n'}{4} + \frac{k - k'}{4}$ edges.}
    \end{array}
    \tag{$\star$}
  \end{equation}

  \smallskip
  \noindent
  {\bf Rule \ref{rule:vertexclique}: }
  Since $v$ is the only vertex connecting $X$ to the rest of the graph, any assignment to $G'$ can be extended to one which is optimal on $X \cup \{v\}$.
  (Let $\alpha$ be an optimal coloring of $G[X \cup \{v\}]$, and let $\alpha'$ be the $\alpha$ with all colors reversed.
  Both $\alpha$ and $\alpha'$ are optimal colorings of $G[X \cup \{v\}]$, and one of these will agree with the coloring we are given on $G'$ since the only overlap is $v$.)
Therefore it suffices to show that  $X+\{v\}$ has a cut of size $\frac{m'}{2} + \frac{n'}{4} + \frac{k - k'}{4}$.
 Observe that $n' = |X|$ and $m' = \frac{|X|(|X|+1)}{2}$, since the edges we remove form a clique including $v$, and all vertices in the clique except $v$ are removed.
  
  If $|X|$ is even then the maximum cut of the clique $X\cup \{v\}$ has size
  $$\frac{|X|}{2} \left(\frac{|X|}{2} +1\right) = \frac{|X|(|X|+2)}{4} = \frac{|X|(|X|+1)}{4} + \frac{|X|}{4} = \frac{m'}{2} + \frac{n'}{4},$$
  which is what we require as $k$ is unchanged in this case.
  
  If $|X|$ is odd then the maximum cut of the clique $X\cup \{v\}$ has size
  $$\frac{(|X|+1)}{2}  \frac{(|X|+1)}{2} = \frac{|X|(|X|+2)}{4}+\frac{1}{4} = \frac{m'}{2} + \frac{n'}{4}  +\frac{1}{4},$$
  which is what we require as we reduce $k$ by~$1$ in this case.

\smallskip
\noindent
 {\bf Rule \ref{rule:vertexmark}: } 
  Order the vertices of $X$ as $x_1, x_2, \hdots, x_{n'}$ such that there exists an $r$ so that~$x_j$ is adjacent to~$v$ for all $j \le r$, and $x_j$ is not adjacent to $v$ for all $j > r$.
  Since~$G$ is connected and reduced by Rule \ref{rule:vertexclique}, $r \ge 2$ (if there is only one vertex $x$ in $X$ adjacent to $v$ then Rule \ref{rule:vertexclique} applies), and since $X$ is a clique but $X \cup \{v\}$ is not (otherwise Rule \ref{rule:vertexclique} applies), there exists $x \in X$ not adjacent to $v$, and so $r\le n'-1$.
  Color the vertices of $X$ such that $x_j$ is the opposite color to $v$ if $j \le \left\lceil\frac{n'}{2}\right\rceil$, and otherwise $x_j$ is the same color as~$v$.
  Observe that the total number of satisfied edges incident with $X$ is $\left\lceil\frac{n'}{2}\right\rceil \left\lfloor\frac{n'}{2}\right\rfloor +\min\left\{r,\left\lceil\frac{n'}{2}\right\rceil\right\}$.
  Since $m' = \frac{n'(n'-1)}{2} + r$, this means that the total number of satisfied edges incident with $X$ is 
  \begin{equation*}
  \frac{m'}{2} + \frac{n'}{4} - \frac{{n'}^2}{4} + \left\lceil\frac{n'}{2}\right\rceil \left\lfloor\frac{n'}{2}\right\rfloor +\min\left\{\frac{r}{2}, \left\lceil\frac{n'}{2}\right\rceil - \frac{r}{2}\right\},
  \end{equation*}
  which is at least $\frac{m'}{2} + \frac{n'}{4} + \frac{1}{2}$ when $|X|$ is even, and at least $\frac{m'}{2} + \frac{n'}{4} + \frac{3}{4}$ when $|X|$ is odd.
  Hence, the number of satisfied edges incident with $X$ is at least $\frac{m'}{2} + \frac{n'}{4} + \frac{1}{2}$.

  \smallskip
  \noindent
  {\bf Rule \ref{rule:triplet}: }
  Observe that $n'=3$ and $m'=2+|E(G',\{a,b,c\})|$.
  Consider two colorings $\alpha, \alpha'$ of $\{a,b,c\}$: $\alpha(a)=\alpha(c) = \mathsf{red}$, $\alpha(b) = \mathsf{blue}$, and $\alpha'(a)=\alpha'(c) = \mathsf{blue}$, $\alpha'(b) =\mathsf{red}$.
  Both these colorings satisfy edges $\{a,b\}$ and $\{b,c\}$, and at least one of them will satisfy at least half the edges between $\{a,b,c\}$ and $G'$.
  Therefore, the number of satisfied edges incident with $\{a,b,c\}$ is at least $2+\frac{|E(G',\{a,b,c\})|}{2} = \frac{m'}{2} + \frac{n'}{4} + \frac{1}{4}$.

 \smallskip
  \noindent
  {\bf Rule \ref{rule:pincer}: }
  Given an assignment to $G'$, color $x,y$ either both $\mathsf{red}$ or both $\mathsf{blue}$, such that at least half of the edges between $G'$ and $\{x,y\}$ are satisfied.
  Assume, without loss of generality, that $x$ and $y$ are both colored $\mathsf{red}$.
  Recall that $X=G  -  (G' \cup \{x,y\})$.
  Let $\bar{n} = |X|$ and let $\bar{m}$ be the number of edges incident with vertices in $X$.  Observe that $\bar{m} = \frac{\bar{n}(\bar{n}-1)}{2} + 2\bar{n}$.

  If $\bar{n}$ is odd, color $\frac{\bar{n}+1}{2}$ of the vertices in $X$ with $\mathsf{blue}$ and $\frac{\bar{n}-1}{2}$ vertices with $\mathsf{red}$. 
  Then the total number of satisfied edges incident with $X$ is $\frac{\bar{n}+1}{2} \cdot \frac{\bar{n}-1}{2} + 2 \frac{\bar{n}+1}{2} = \frac{\bar{n}(\bar{n}-1)}{4} + \frac{\bar{n}-1}{4} + \bar{n} + 1 =  \frac{\bar{m}}{2} + \frac{\bar{n}-1}{4}+1 = \frac{\bar{m}}{2}+\frac{\bar{n}}{4}+\frac{3}{4}$.

  If $\bar{n}$ is even, color $\frac{\bar{n}}{2}+1$ of the vertices in $X$ with $\mathsf{blue}$ and $\frac{\bar{n}}{2}-1$ vertices with $\mathsf{red}$.
  Then the total number of satisfied edges incident with $X$ is $(\frac{\bar{n}}{2}+1)(\frac{\bar{n}}{2}-1) + 2(\frac{\bar{n}}{2}+1) = \frac{{\bar{n}}^2}{4} -1 + \bar{n} + 2 = \frac{\bar{m}}{2} + \frac{\bar{n}}{4} + 1$.

  Observe that $m' = |E(G',\{x,y\})| + \bar{m}$ and $n' = \bar{n}+2$, and hence the total number of removed edges satisfied is at least $\frac{|E(G',\{x,y\})|}{2} + \frac{\bar{m}}{2} + \frac{\bar{n}}{4} + \frac{3}{4} = \frac{m'}{2} + \frac{n'-2}{4} + \frac{3}{4} = \frac{m'}{2} + \frac{n'}{4} + \frac{1}{4}$. 

\smallskip

This concludes the proof of the claim~\eqref{eqn:starclaim}.
  
  We now know that any assignment on~$G'$ can be extended to an assignment on~$G$ that cuts an additional $\frac{m'}{2} + \frac{n'}{4} + \frac{k - k'}{4}$ edges.
  Hence, if $G'$ has a cut of size $\frac{|E(G')|}{2} + \frac{|V(G')|-1}{4} + \frac{k'}{4}$, then $G$ has a cut of size $\frac{m-m'}{2} + \frac{n-n'-1}{4} + \frac{k'}{4} + \frac{m'}{2} + \frac{n'}{4} + \frac{k - k'}{4} = \frac{m}{2} + \frac{n-1}{4} + \frac{k}{4}$.
  Therefore, if $(G',k')$ is a ``yes''-instance then so is $(G,k)$.
\end{proof}

\begin{lemma}
  \label{thm:exhaustion}
  To any connected graph $G$ with at least one edge, at least one of Rules~\ref{rule:vertexclique}--\ref{rule:pincer} applies.
\end{lemma}
\begin{proof}
If $G$ is not $2$-connected, then let $X$ be a leaf-block in $G$ with cut-vertex $r$.
Otherwise, let $X = G$ and let $r$ be any vertex in $G$.
  We say that a set of vertices $\{a,b,c\}$ in $G$ is an \emph{induced $P_3$} if $G[\{a,b,c\}]$ is a path with $3$ vertices (i.e. $ab, bc \in E(G), ac \notin E(G)$).
  For any set $S$ of vertices in $X$, let $R(r,S)$ denote the set of vertices reachable from $r$ in $G - S$ (if $r \in S$ then $R(r,S) = \emptyset$) . 
  For $a,b,c \in X$, we say $\{a,b,c\}$ is an \emph{important $P_3$} if $\{a,b,c\}$ is an induced $P_3$ and there is no induced $P_3$ $\{a',b',c'\} \subseteq X$ such that $R(r, \{a,b,c\}) \subset R(r,\{a',b',c'\})$.
  (The concept of an important $P_3$ is similar to the concept of ``important separators'' introduced by Marx \cite{Marx2006}.)
  
  Observe that if there is no induced $P_3$ in $X$ then $X$ is a clique, and Rule \ref{rule:vertexclique} applies. Therefore we may assume there is an important $P_3$ in $X$. Let $\{a,b,c\}$ be an important $P_3$ with $ab, bc \in E(G), ac \notin E(G)$.
  If $r \in \{a,b,c\}$ then $X - r$ must have no induced $P_3$ and so $X - r$ is a clique, and Rule \ref{rule:vertexmark} applies. Therefore we may assume $r \notin \{a,b,c\}$.

  Let $Y = G - (R(r,\{a,b,c\})\cup \{a,b,c\})$. If $Y=\emptyset$ then $G - \{a,b,c\}$ is connected and Rule \ref{rule:triplet} applies. Therefore we may assume $Y \neq \emptyset$.
  Since $X$ is $2$-connected, at least two of $a,b,c$ are adjacent to $R(r,\{a,b,c\})$. In particular, one of $a$ and $c$ is adjacent to $R(r,\{a,b,c\})$.
  Without loss of generality, assume that~$c$ is adjacent to $R(r,\{a,b,c\})$.

  We now prove three properties of $Y$, which will be used to show that Rule \ref{rule:pincer} applies:
  \begin{enumerate}[label=(P\arabic{*}), ref=(P\arabic{*})]
    \item\label{prop:1} For any $x \in Y$, $ax \in E(G)$ if and only if $bx \in E(G)$.\\
      For suppose there exists $x \in Y$ which is adjacent to $a$ but not $b$.
      Then $\{x,a,b\}$ is an induced $P_3$, and $R(r,\{x,a,b\}) \supseteq R(r,\{a,b,c\}) \cup \{c\}$, a contradiction as $\{a,b,c\}$ is an important $P_3$.
      Similarly, if $x$ is adjacent to $b$ but not $a$, then $\{a,b,x\}$ is an induced $P_3$ and $R(r,\{a,b,x\}) \supseteq R(r,\{a,b,c\}) \cup \{c\}$.
    \item\label{prop:2} For each $s \in \{a,b,c\}$ and any $x,y \in Y$, if $sx \in E(G)$ and $xy \in E(G)$, then $sy \in E(G)$.\\
      For suppose not; then $\{s,x,y\}$ is an induced $P_3$, and since at least one vertex $t \in \{a,b,c\}\backslash s$ is adjacent to $R(r,\{a,b,c\})$, we have that $R(r,\{s,x,y\}) \subseteq R(r,\{a,b,c\}) \cup \{t\}$, a contradiction as $\{a,b,c\}$ is an important $P_3$.
    \item\label{prop:3} For each $s \in \{a,b,c\}$ and any $x,y \in Y$, if $sx \in E(G)$ and $sy \in E(G)$, then $xy \in E(G)$.\\
      For suppose not. Then $\{x,s,y\}$ is an induced $P_3$, and as with Property~\ref{prop:2} we have a contradiction.
  \end{enumerate}

  There are now two cases to consider.
  First, consider the case when $a$ is adjacent to $R(r,\{a,b,c\})$. 
  Then, by an argument similar to that used for proving Property~\ref{prop:1}, we may show that for any $x \in Y$, $cx \in E(G)$ if and only if $bx \in E(G)$.
  This together with Property~\ref{prop:1} implies that $a,b,c$ have exactly the same neighbors in $Y$. 
  Then by Properties~\ref{prop:2} and~\ref{prop:3} we have that $Y$ is a clique and every vertex in $Y$ is adjacent to each of $a,b,c$.
  If $b$ is adjacent to $R(r,\{a,b,c\})$ then for any $x \in Y$, $\{a,x,c\}$ is an induced $P_3$ and $R(r, \{a,x,c\}) \supseteq R(r,\{a,b,c\}) \cup \{b\}$, a contradiction as $\{a,b,c\}$ is an important $P_3$. So $b$ is not adjacent to $R(r,\{a,b,c\})$. 
  Then  $Y \cup \{b\}$  and $R(r,\{a,b,c\})$ are the two components of $G - \{a,c\}$, and $(Y \cup \{b\}) \cup\{a\}$ and $(Y \cup \{b\}) \cup \{c\}$ are cliques.
  Therefore, Rule \ref{rule:pincer} applies.

  Second, consider the case when $a$ is not adjacent to $R(r,\{a,b,c\})$.
  Then $b$ must be adjacent to $R(r,\{a,b,c\})$.
  Furthermore, as $X$ is $2$-connected, there must be a path from $a$ to $c$ in $G - b$, and the intermediate vertices in this path must be in $Y$.
  By Property~\ref{prop:2}, there exists $x \in Y$ adjacent to $a$ and $c$.
  Then $\{a,x,c\}$ is an induced $P_3$ and $R(r, \{a,x,c\}) \supseteq R(r,\{a,b,c\}) \cup \{b\}$, a contradiction as $\{a,b,c\}$ is an important $P_3$.
\end{proof}

\begin{lemma}
\label{thm:treeOfCliques}
  Let $(G',k')$ be an instance derived from $(G,k)$ by exhaustively applying Rules \ref{rule:vertexclique}--\ref{rule:pincer}, and let $S$ be the set of vertices marked during the construction of $(G',k')$. Then $G - S$ is a clique-forest.
\end{lemma}
\begin{proof}
  By Lemma \ref{thm:exhaustion}, $G'$
  contains no edges, and is therefore a clique-forest. 
  (In fact by Lemma \ref{thm:validity}, $G'$ is also connected, and therefore consists of a single vertex.)
  Let~$G''$ be a graph derived from $G$ by a single application of Rule \ref{rule:vertexclique}, \ref{rule:vertexmark}, \ref{rule:triplet} or \ref{rule:pincer}.
  By induction on the length of the reduction from $(G,k)$ to $(G',k')$, it is enough to show that if $G'' - S$ is a clique-forest, then $G - S$ is a clique-forest.

  If $G''$ is derived from $G$ by an application of Rule \ref{rule:vertexclique}, observe that $G -  S$ can be formed from~$G''  -  S$ by adding a disjoint clique and identifying one of its vertices with the vertex~$v$ in $G''$ (unless $v$ is in $S$, in which case $G -  S$ is just formed by adding a disjoint clique), where $v$ is the vertex referred to in Rule~\ref{rule:vertexclique}.
Therefore, $G  -  S$ is a clique-forest. 
For Rule \ref{rule:triplet}, observe that $G - S = G'' - S$.

  For Rules \ref{rule:vertexmark} and \ref{rule:pincer}, observe that $G  -  (G' \cup S)$ is a clique,
  and that $S$ disconnects $G  -  (G' \cup S)$ from $G'$.
  Therefore, $G  -  S$ can be formed from $G'  -  S$ by adding a disjoint clique, and so $G  -  S$ is a clique-forest.
\end{proof}

Putting Lemmas \ref{thm:validity} and \ref{thm:treeOfCliques} together, we can now prove Lemma \ref{thm:findS}.

\begin{proof}[of Lemma \ref{thm:findS}]
Let $(G',k')$ be an instance derived from $(G,k)$ by exhaustively applying Rules~\ref{rule:vertexclique}--\ref{rule:pincer}, and let $S$ be the set of vertices marked during the construction of $(G',k')$.
  By Lemma \ref{thm:treeOfCliques}, $G  -  S$ is a clique-forest.
  Therefore, if $|S|<3k$ we are done.
  It remains to show that if $|S|\ge 3k$ then $G$ has an assignment that satisfies at least $\frac{m}{2} + \frac{n-1}{4} + \frac{k}{4}$ edges.

  So suppose that $|S|\ge 3k$. 
  Observe that every time $k$ is reduced, at most three vertices are marked.
  Therefore since at least $3k$ vertices are marked, we have $k'\le 0$.
  But since the Edwards-Erd\H{o}s bound holds for all connected graphs, $G'$ is a ``yes''-instance. Therefore, by Lemma \ref{thm:validity}, $G$ is a ``yes''-instance with parameter $k$, as required.
\end{proof}

In what follows, we assume we are given a set $S$ such that $|S|\le 3k$ and $G - S$ is a clique-forest, as given by Lemma \ref{thm:findS}. Note that we should not assume that $(G,k)$ is reduced by any of Rules \ref{rule:vertexclique}--\ref{rule:pincer}. These rules were only used to produce the separate instance $(G',k')$ and $S$. We may now forget about $(G',k')$ and Rules \ref{rule:vertexclique}--\ref{rule:pincer}.

We now show that, for a given assignment to $S$, we can efficiently find an optimal extension to $G  -  S$.
For this, we consider the following generalisation of {\sc Max-Cut} where each vertex has an associated weight for each part of the partition.
These weights may be taken as an indication of how much we would like the vertex to appear in each part.
For convenience, here we will think of an assignment as a function from $V$ to $\{0,1\}$ rather than $\{\mathsf{red}, \mathsf{blue}\}$.

\begin{center}
\fbox{~\begin{minipage}{0.8\textwidth}
\textsc{Max-Cut-with-Weighted-Vertices}

\emph{Instance}: A graph $G$ with weight functions $w_0:~V(G)~\rightarrow~\mathbb{N}_0$ and $w_1:~V(G)~\rightarrow~\mathbb{N}_0$, and an integer $t\in\mathbb N$.

\smallskip

\emph{Question}: Does there exist an assignment $f: V \rightarrow \{0,1\}$ such that\\
$\sum_{xy \in E}|f(x)-f(y)| + \sum_{f(x)=0}w_0(x) + \sum_{f(x)=1}w_1(x) \ge t$?
\end{minipage}~}
\end{center}

Now \textsc{Max-Cut} is the special case of \textsc{Max-Cut-with-Weighted-Vertices} in which
 $w_0(x)=w_1(x)=0$ for all $x \in V(G)$.

\begin{lemma}
\label{thm:poly}
  \textsc{Max-Cut-with-Weighted-Vertices} is solvable in polynomial time when $G$ is a 
clique-forest.
\end{lemma}

\begin{proof}
  We provide a polynomial-time transformation that replaces an instance $(G,w_0,w_1,t)$ with an equivalent instance $(G', w_0', w_1', t')$ such that $G'$ has fewer vertices than $G$. By applying the transformation at most $|V(G)|$ times to get a trivial instance, we have a polynomial-time algorithm to decide $(G,w_0,w_1,t)$.

  We may assume that $G$ is connected, as otherwise we can handle each component of $G$ separately.
  Let $X \cup \{r\}$ be the vertices of a leaf-block in $G$, with $r$ a cut-vertex of $G$ (unless $G$ consists of a single block, in which case let $r$ be an arbitrary vertex and $X  = V(G) - \{r\}$). Recall that by definition of a clique-forest, $X \cup \{r\}$ is a clique.
  For each possible assignment to $r$, we will in polynomial time calculate the optimal extension to the vertices in $X$. (This optimal extension depends only on the assignment to~$r$, since no other vertices are adjacent to vertices in $X$.)
  We can then remove all the vertices in $X$, and change the values of $w_0(r)$ and $w_1(r)$ to reflect the optimal extension for each assignment.

  Suppose we assign $r$ the value $1$.
  Let $\varepsilon(x) = w_1(x) - w_0(x)$ for each $x \in X$.
  Now arrange the vertices of $X$ in order $x_1, x_2, \hdots x_{n'}$ (where $n'=|X|$), such that if $i < j$ then $\varepsilon(x_i) \ge \varepsilon(x_j)$.
  Observe that there is an optimal assignment for which $x_i$ is assigned $1$ for every $i\le t$, and $x_i$ is assigned $0$ for every $i > t$, for some $0 \le t \le n'$. (Consider an assignment for which $f(x_i)=0$ and $f(x_j)=1$, for $i<j$, and observe that switching the assignments of $x_i$ and $x_j$ will 
increase $\sum_{f(x)=0}w_0(x) + \sum_{f(x)=1}w_1(x)$ by $\varepsilon(x_i) - \varepsilon(x_j)$.)
  Therefore we only need to try $n'+1$ different assignments to the vertices in $X$ in order to find the optimal coloring when $f(r)=1$.
  Let $A$ be the value  of this optimal assignment (over $X \cup \{r\}$).

  By a similar method we can find the optimal assignment when $r$ is assigned $0$. Let the number of satisfied edges in this coloring be $B$.
  Now remove the vertices in $X$ and incident edges, and let $w_1(r) = A$, and let $w_0(r) = B$.
\end{proof}


We are ready to prove Theorem~\ref{thm:maxcutaboveeeisfpt}, and show that {\sc Max-Cut-AEE} is fixed-parameter tractable.

\begin{proof}[of Theorem \ref{thm:maxcutaboveeeisfpt}]
  By Lemma \ref{thm:findS}, we can in polynomial time decide that either $G$ has an assignment that satisfies at least $\frac{m}{2} + \frac{n-1}{4} + \frac{k}{4}$ edges, or find a set $S$ of at most $3k$ vertices in $G$ such that -$G  -  S$ is a clique-forest.
  So assume we have found such an~$S$. 
  Then we transform our instance into at most $2^{3k}$ instances of \textsc{Max-Cut-with-Weighted-Vertices}, such that the answer to our original instance is ``yes'' if and only if the answer to at least one of the instances of \textsc{Max-Cut-with-Weighted-Vertices} is ``yes'', and in each \textsc{Max-Cut-with-Weighted-Vertices} instance the graph is a clique-forest.
  As each of these instances can be solved in polynomial time by Lemma \ref{thm:poly}, we have a fixed-parameter tractable algorithm.

  For every possible assignment to the vertices in $S$, we construct an instance of \textsc{Max-Cut-with-Weighted-Vertices} as follows.
  For every vertex $x \in G  -  S$, let $w_0(x)$ equal the number of vertices in~$S$ adjacent to $x$ which are colored $\mathsf{blue}$, and let $w_1(x)$ equal the number of vertices in~$S$ adjacent to~$x$ which are colored $\mathsf{red}$. Then remove the vertices of~$S$ from $G$.
  By Lemma \ref{thm:findS}, the resulting graph $G'$ is a 
clique-forest.
  Let~$m'$ be the number of edges in $G  -  S$, let $n'$ be the number of vertices in $G -  S$, and let $p$ be the number of edges within $S$ satisfied by the assignment to $S$.
  Then for an assignment to the vertices of $G  -  S$, the total number of satisfied edges in $G$ would be exactly $\sum_{xy \in E(G -  S)}|f(x)-f(y)| + p + \sum_{f(x)=0}w_0(x) + \sum_{f(x)=1}w_1(x)$, where $f: V(G)\setminus  S \rightarrow \{0,1\}$ is such that $f(x)=0$ if $x$ is colored $\mathsf{red}$, and $f(x)=1$ if $x$ is colored $\mathsf{blue}$.
  Thus, the assignment to $S$ can be extended to one that cuts at least $\frac{m}{2} + \frac{n-1}{4} + \frac{k}{4}$ edges in~$G$ if and only if the instance of \textsc{Max-Cut-with-Weighted-Vertices} is a ``yes''-instance with 
$t = \frac{m}{2} + \frac{n-1}{4} + \frac{k}{4} - p$.
\end{proof}

\section{Algorithmic Lower Bounds}
\label{sec:algorithmiclowerbounds}
We now prove Theorem~\ref{thm:ethlowerbound}, by a reduction from the {\sc Max-Cut} problem with parameter the size $k$ of the cut.
By a result of Cai and Juedes~\cite{CaiJuedes2003}, the {\sc Max-Cut} problem with parameter the size $k$ of the cut cannot be solved in time $2^{o(k)}\cdot n^{O(1)}$ unless the Exponential Time Hypothesis fails.
Note that by deleting isolated vertices and identifying together one vertex from each component, we may assume the input graph of an instance of {\sc Max Cut} is connected.
Given the graph $G$ with $m$ edges and $n$ vertices and the integer $k\in\mathbb N$, 
let $k' = k - \lceil m/2 + (n-1)/4 \rceil$. That is, $k'$ is the integer such that 
$G$ has a cut of size at least $k$ if and only if $G$ has a cut of size at least $m/2 + (n-1)/4 + k'$.
Then use a hypothetical algorithm to solve {\sc Max-Cut AEE} in time $2^{o(k')}\cdot n^{O(1)}$ on instance $(G,k')$, return the answer of the algorithm for the pair $(G,k)$ and the {\sc Max-Cut} problem.
Thus, since $k'\leq k$, an algorithm of time complexity $2^{o(k)}\cdot n^{O(1)}$ for {\sc Max-Cut AEE} forces the Exponential Time Hypothesis to fail~\cite{DowneyFellows1999}.
This completes the proof of Theorem~\ref{thm:ethlowerbound}.

\section{Polynomial Kernel for Max-Cut Above \texorpdfstring{Edwards-Erd\H{o}s}{Edwards-Erdos}}
\label{sec:polykernel}
In this section, we prove Theorem \ref{thm:polykernelstrongerbound}.
By Lemma \ref{thm:findS}, in polynomial time we can either decide that $(G,k)$ is a ``yes''-instance, or find a set $S$ of vertices in $G$ such that $|S|<3k$ and 
$G  -  S$ is a clique-forest.
In what follows we assume that we have found such a set $S$.

Let $n^*$ be the number of blocks of $G - S$. Let $C_1, \hdots, C_{n^*}$ be the blocks of $G - S$, and recall that by definition each block is a clique. Let $J$ be the set of vertices in $G - S$ that appear in two or more block. For each block $C_i$, let $A_i = C_i \setminus J$. Note that $\{A_1, \hdots , A_{n^*}, J\}$ is a partition of the vertices of $G  -  S$.
Let $L \subseteq \{1,\hdots,n^*\}$ be the set of all $i$ for which $C_i$ is a leaf-block.
Let $I \subseteq [n^*]$ be the set of all $i$ for which $|C_i \cap J|\ge 3$.
It can be seen that $|L| \ge |I|+2$.  (Indeed, consider the bipartite graph with vertex sets $J$ and $[n^*]$, with an edge between $v \in J$ and $i \in [n^*]$ if $v \in C_i$. Then this is a forest, with all leaves being leaf-blocks in $G-S$, and for each $i \in [n^*]$ the degree of $i$ is $|C_i \cap J|$. It is well known that in a tree the number of vertices of degree at least three is at most the number of leaves $ - 2$, from which the claim follows.)

We first apply Reduction Rules \ref{rule:kerRule1}--\ref{rule:kerRule4}. Note that unlike Rules~\ref{rule:vertexclique}--\ref{rule:pincer}, these are traditional ``two-way'' reduction rules, and once we have shown they are valid, we may assume the instance is reduced by these rules.
The correctness of these rules is given by Lemma \ref{thm:validity2}.

After applying these reduction rules, we will show that either we have a {\sc Yes}-instance, or the number of vertices in $G-S$ is bounded. Observe that to bound the number of vertices in $G-S$, it is enough to bound $n^*$ and $|A_i|$ for each $i$. In Lemma \ref{thm:noLeafCliques}, we bound $|L|$. As $|L| \ge |I|+2$, in order to bound $n^*$ it remains to bound the number of $i$ for which  $|C_i \cap J|\ge 3$. This is done in Lemma \ref{thm:noCliques}. Finally in Lemma \ref{thm:cliqueSize} we bound $|A_i|$ for each $i$. 
These results are combined in Lemma \ref{thm:noVertices}.

\smallskip
\refstepcounter{krule}
\label{rule:kerRule1}
\noindent
\framebox[\textwidth][l]{
  \begin{tabularx}{0.98\textwidth}{lX}
    \krulehead:     & Apply if there exists a vertex $x \in G  -  S$ and a set of vertices $X \subseteq V(G)\setminus S$ such that $|X|>1$,
 $X\cup \{x\}$ is a clique and $X$ is a connected component of $G  -  (S \cup \{x\})$, and no vertex in $X$ is adjacent to any vertex in $S$.\\
    Remove:        & All vertices in $X$.\\
    Parameter:     & Reduce $k$ by $1$ if $|X|$ is odd, otherwise leave $k$ the same.\\
  \end{tabularx}
}

\smallskip
\refstepcounter{krule}
\label{rule:kerRule2}
\noindent
\framebox[\textwidth][l]{
  \begin{tabularx}{0.98\textwidth}{lX}
    \krulehead:     & Apply if there exists vertices $s \in S,x \in G  -  S$ and a set of vertices $X \subseteq V(G)\setminus S$ such that  $X\cup \{x\}$ is a clique,  $X \cup \{s\}$ is a clique, $X$ is a connected component of $G  -  (S \cup \{x\})$, and $s$ is the only vertex in $S$ adjacent to $X$.\\
    Remove:        & All but one vertex of $X$.\\
    Parameter:     & Reduce $k$ by $1$ if $|X|$ is even, otherwise leave $k$ the same.\\
  \end{tabularx}
}

\smallskip
\refstepcounter{krule}
\label{rule:kerRule3}
\noindent
\framebox[\textwidth][l]{
  \begin{tabularx}{0.98\textwidth}{lX}
    \krulehead:     & Apply if there exist blocks $X, Y$ of $G  -  S$ such that $|X|$ and $|Y|$ are odd, with vertices $x \in X, y \in Y, \{z\} = X\cap Y$, such that $x,z$ are the only vertices in $X$ adjacent to a vertex in $G -  X$, and $y,z$ are the only vertices in $Y$ adjacent to a vertex in $G  -  Y$.\\
    Remove:        & All vertices in $(X \cup Y) -  \{x,y,z\}$.\\
    Add:           & New vertices $u,v$ and edges such that $\{x,y,z,u,v\}$ is a clique\\
    Parameter:     & No change.\\
  \end{tabularx}
}

\smallskip
\refstepcounter{krule}
\label{rule:kerRule4}
\noindent
\framebox[\textwidth][l]{
  \begin{tabularx}{0.98\textwidth}{lX}
    \krulehead:    & Apply if for any block $C_i$ in $G  -  S$, there exists $X \subseteq A_i $ such that $|X| > \frac{|A_i|+|J|+|S|}{2}$ and for all $x,y \in X$, $x$ and $y$ have exactly the same neighbors in $S$.\\
    Remove:        & Two arbitrary vertices from $X$.\\
    Parameter:     & No change.\\
  \end{tabularx}
}


\begin{lemma}
\label{thm:validity2}
   Let $(G,k)$ and $(G',k')$ be instances of \textsc{Max-Cut-AEE}, and $S$ a set of vertices, such that $G  -  S$ is a clique-forest,
  and $(G',k')$ is reduced from $(G,k)$ by an application of Rule \ref{rule:kerRule1}, \ref{rule:kerRule2}, \ref{rule:kerRule3} or \ref{rule:kerRule4}.
  Then $G'$ is connected, $G'  -  S$ is a clique-forest,
  and $(G',k')$ is a ``yes''-instance if and only if $(G,k)$ is a ``yes''-instance.
\end{lemma}
\begin{proof}[of Lemma \ref{thm:validity2}]
  It is easy to observe that each of the rules preserves connectedness, and that $G'  -  S$ is a clique-forest.

  We now show for each rule that $(G',k')$ is a ``yes''-instance if and only if $(G,k)$ is a ``yes''-instance.

  \smallskip
  \noindent
  {\bf Rule \ref{rule:kerRule1}:}
  Let $n'$ be the number of vertices and $m'$ the number of edges removed.
  Note that $m'=\frac{n'(n'+1)}{2}$.
  Observe that whatever we assign to the rest of the graph, we can always find an assignment to $X$ that satisfies the largest possible number of edges within $X \cup \{x\}$. If $|X|$ is odd this is $\frac{(n'+1)(n'+1)}{4} = \frac{m'}{2}+\frac{n'+1}{4}$, and if $|X|$ is even this is $\frac{n'(n'+2)}{4} = \frac{m'}{2}+\frac{n'}{4}$.
  Therefore if $|X|$ is odd, we can satisfy 
  $\frac{m}{2} + \frac{n-1}{4} + \frac{k}{4}$ edges in $G$ if and only if we can satisfy $\frac{m-m'}{2} + \frac{n-n'-1}{4} + \frac{k-1}{4}$ edges in $G'$, and if $|X|$ is even, we can satisfy  $\frac{m}{2} + \frac{n-1}{4} + \frac{k}{4}$ edges in $G$ if and only if we can satisfy $\frac{m-m'}{2} + \frac{n-n'-1}{4} + \frac{k}{4}$ edges in $G'$.

  \smallskip
  \noindent
  {\bf Rule \ref{rule:kerRule2}:}
 Let $n'$ be the number of vertices and $m'$ the number of edges removed.
  Note that $n' = |X|-1$ and $m'=\frac{n'(n'+1)}{2}+2n' = \frac{n'(n'+5)}{2}$.

First consider the case when $x$ and $s$ are adjacent. Observe that whatever $x$ and $s$ are assigned, it is possible to find an assignment to $X$ that satisfies the maximum possible number of edges within $X \cup \{s,x\}$.
This is $\frac{(n'+2)(n'+4)}{4} = \frac{m'}{2} + \frac{n'+8}{4}$ if $|X|$ is odd, and $\frac{(n'+3)(n'+3)}{4} = \frac{m'}{2} + \frac{n'+9}{4}$ if $|X|$ is even. Furthermore note that in $G'$, whatever $x$ and $s$ are assigned, we will be able to satisfy $2$ edges between $x,s$ and the reamining part of $X$.
  Therefore if $|X|$ is odd, we can satisfy 
  $\frac{m}{2} + \frac{n-1}{4} + \frac{k}{4}$ edges in $G$ if and only if we can satisfy $\frac{m-m'}{2} + \frac{n-n'-1}{4} + \frac{k-8}{4} + 2 = \frac{m-m'}{2} + \frac{n-n'-1}{4} + \frac{k}{4}$ edges in $G'$, and if $|X|$ is even, we can satisfy  $\frac{m}{2} + \frac{n-1}{4} + \frac{k}{4}$ edges in $G$ if and only if we can satisfy $\frac{m-m'}{2} + \frac{n-n'-1}{4} + \frac{k-9}{4} + 2 = \frac{m-m'}{2} + \frac{n-n'-1}{4} + \frac{k-1}{4}$ edges in $G'$.

%
  
%
%

  Second consider the case when $x$ and $s$ are not adjacent.
  Observe that if $x$ and~$s$ are colored differently, the maximum number of edges within $X \cup \{x, s\}$ we can satisfy is $\frac{(n'+2)(n'+4)}{4} -1 = \frac{m'}{2} + \frac{n'+4}{4}$ if $|X|$ is odd, and $\frac{(n'+3)(n'+3)}{4} -1 = \frac{m'}{2} + \frac{n'+5}{4}$ if $|X|$ is even. Furthemore in $G'$ we will be able to satisfy $1$ edge between $x,s$ and the remaining part of $X$.

  If $x$ and $s$ are colored the same, the maximum number of edges within $X \cup \{x,s\}$ we can satisfy is $\frac{(n'+2)(n'+4)}{4} -1 = \frac{m'}{2} + \frac{n'+8}{4}$  if~$|X|$ is odd, and $\frac{(n'+3)(n'+3)}{4} -1 = \frac{m'}{2} + \frac{n'+9}{4}$ if $|X|$ is even.
Furthemore in $G'$ we will be able to satisfy $2$ edges between $x,s$ and the remaining part of $X$.

Observe that whether or not $x$ and $s$ are colored the same, in $G$ we can satisfy $t$ more edges between $x,s$ and $X$ than in $G'$, where $t = \frac{m-m'}{2} + \frac{n-n'-1}{4} + \frac{k}{4}$ if $|X|$ is odd, and $t = \frac{m-m'}{2} + \frac{n-n'-1}{4} + \frac{k-1}{4}$ if $|X|$ is even.


%

  \smallskip
  \noindent
  {\bf Rule \ref{rule:kerRule3}: }
  Let $n' = |X\cup Y|$ and let $m'$ be the number of edges within $X\cup Y$.
  Let $n'_1 = |X|$ and $n'_2=|Y|$, and let $m'_1$ and $m'_2$ be the number of edges within $X$ and $Y$, respectively.
  Observe that whatever $x,y$ and $z$ are assigned, we can always satisfy the maximal possible number of edges in within $X\cup Y$, which is 
  \begin{equation*}
  \frac{(n'_1+1)(n'_1-1)}{4} + \frac{(n'_2+1)(n'_2-1)}{4} = \frac{m'_1}{2}+\frac{m'_2}{2}+\frac{n'_1-1}{4}+\frac{n'_2-1}{4} = \frac{m'}{2}+\frac{n'-1}{4} \enspace .
  \end{equation*}

  Let $n'_3=5=|\{x,y,z,u,v\}|$ and let $m'_3=10$ be the number of edges within $\{x,y,z,u,v\}$.
  Then in~$G'$, whatever $x,y$ and $z$ are assigned, the maximum number of edges within $\{x,y,z,u,v\}$ we can satisfy is $6=\frac{m'_3}{2}+\frac{n'_3-1}{4}$.
  Thus, the amount we gain above the Edwards-Erd\H{o}s bound remains the same.

  \smallskip
  \noindent
  {\bf Rule \ref{rule:kerRule4}: }
  Let $n'=|A_i|$.
  For any assignment to the vertices in $S\cup J$, and for each $x \in A_i$, let $\varepsilon_R(x)$ be the number of neighbors of $x$ in $S\cup J$ which are assigned $\mathsf{red}$, and let $\varepsilon_B(x)$ be the number of neighbors of $x$ in $S\cup J$ which are assigned $\mathsf{blue}$.
  Let $\varepsilon(x)=\varepsilon_B(x)-\varepsilon_R(x)$.
  Let the vertices of $A_i$ be numbered $x_1, x_2, \hdots , x_{n'}$ such that $\varepsilon(x_1)\ge \varepsilon(x_2)\ge \hdots \ge \varepsilon(x_{n'})$.
  Observe that the optimal assignment to $A_i$ will be one in which $x_j$ is assigned $\mathsf{red}$ for $j \le \frac{n'+r}{2}$, and $\mathsf{blue}$ otherwise, for some integer $r$. Observe that the optimal value of $r$ will be between $-(|J|+|S|)$ and $(|J|+|S|)$.
  Indeed, if $r>|J|+|S|$, then switching one of the vertices from $\mathsf{red}$ to $\mathsf{blue}$ will gain at least $|J|+|S|$ edges within $A_i$, and lose at most $|J|+|S|$ edges between $A_i$ and~$J \cup S$.
  (A similar argument holds when $r<-|J|-|S|$.)

  Now let $y,z$ be the two vertices in $X$ removed by the rule. Since $y,z$ have exactly the same neighbors in $S \cup J$, and since $|X| > \frac{|A_i|+|J|+|S|}{2}$, we may assume that $y=x_j, z=x_{j'}$, for some $j,j'$ such that $j'-j > |J \cup S| \ge r$, and hence $y$ is assigned $\mathsf{red}$  and  $z$ is assigned $\mathsf{blue}$.
  Then note that if we remove $y$ and $z$, we lose $m''= 2|N(y)\cap (S\cup J)| + 2(n'-2)+1$ edges and $n''=2$ vertices.
  Of the edges removed, exactly $|N(y)\cap (S\cup J)| + (n'-2)+1 = \frac{m''}{2}+\frac{n''}{4}$ were satisfied.
  Thus, the amount we gain over the Edwards-Erd\H{o}s bound remains the same.
  Note that this happens whatever the assignment to $S \cup J$, and that we may assume without loss of generality that $y$ and $z$ are colored differently for any assignment to~$S \cup J$.
\end{proof}

We now assume that $G$ is reduced by Rules \ref{rule:kerRule1}, \ref{rule:kerRule2}, \ref{rule:kerRule3} and~\ref{rule:kerRule4}.

\begin{lemma}
\label{thm:noLeafCliques}
  If $|L|\ge 4|S|^2+2|S|+2k-2$ then $(G,k)$ is a ``yes''-instance.
\end{lemma}
\begin{proof}
  Suppose $|L| \ge 4|S|^2+2|S|+2k-2$.
  For a leaf block $C_i$, let $E_i$ be set of edges within $G[A_i]$ together with the set of edges between $A_i$ and $S$, and let $E_L=\bigcup_{i \in L}E_i$.
  Let $V_L$ be the total set of vertices in $\bigcup_{i\in L}A_i$, and observe that $|V_L|=\sum_{i \in L}|A_i|$. Let $m_i=|E_i|$ and $n_i=|A_i|$.

  For a partial assignment $\alpha$ on $L$ and a set of edges $E$, let $\chi_\alpha(E)$ be the maximum possible number of edges we can satisfy in $E$, given $\alpha$.
  We will show via the probabilistic method that there exists a partial assignment $\alpha$ on $S$ such that  $\chi_\alpha(E_L) \ge \frac{|E_L|}{2} + \frac{|V_L|}{4} + \frac{|L|}{4}$.
  Consider a random assignment $\alpha$ on $S$, in which each vertex is assigned $\mathsf{red}$ or $\mathsf{blue}$ with equal probability.
  For any vertex $x \in G  -  S$, let $\varepsilon^\alpha_R(x)$ be the number of neighbors of $x$ which are assigned $\mathsf{red}$, let $\varepsilon^\alpha_B(x)$ be the number of neighbors of $x$ which are assigned $\mathsf{blue}$, and let $\varepsilon^\alpha(x) = \varepsilon^\alpha_B(x)-\varepsilon^\alpha_R(x)$.
  Now consider $\chi_\alpha(E_i)$ for some $i$ where $i \in L$.
  Let $x_1, \hdots, x_{n_i}$ be an ordering of the vertices of $A_i$ such that $\varepsilon^\alpha(x_1) \ge \varepsilon^\alpha(x_2) \ge \hdots \ge \varepsilon^\alpha(x_{n_i})$. 

  First, suppose that $|A_i|$ is odd.
  Observe that if there exists $x \in A_i$ with $|\varepsilon^\alpha(x)| \ge r$, then $\chi_\alpha(E_i) \ge \frac{m_i}{2}+\frac{n_i-1}{4} + \frac{r}{2}$.
  Indeed, suppose $\varepsilon^\alpha(x)=r'>r$. Then color $x_j$ $\mathsf{red}$ if $j \le \frac{n_i+1}{2}$ and $\mathsf{blue}$ otherwise.
  Then the total number of satisfied edges in $E_i$ is equal to
  \begin{multline*}
    \frac{n_i+1}{2}\cdot\frac{n_i-1}{2} + \sum_{j=1}^{\frac{n_i+1}{2}}\varepsilon^\alpha_B(x_j) + \sum_{j=\frac{n_i+3}{2}}^{n_i}\varepsilon^\alpha_R(x_j) \\
     =  \frac{n_i(n_i-1)}{4} + \frac{1}{2}\sum_{j=1}^{\frac{n_i+1}{2}}(\varepsilon^\alpha_B(x_j) + \varepsilon^\alpha_R(x_j)) + \frac{1}{2}\sum_{j=\frac{n_i+3}{2}}^{n_i}(\varepsilon^\alpha_B(x_j)+ \varepsilon^\alpha_R(x_j))\\
      + \frac{n_i-1}{4} + \frac{1}{2} \sum_{j=1}^{\frac{n_i+1}{2}}(\varepsilon^\alpha_B(x_j) - \varepsilon^\alpha_R(x_j)) + \frac{1}{2}\sum_{j=\frac{n_i+3}{2}}^{n_i}(\varepsilon^\alpha_R(x_j) - \varepsilon^\alpha_B(x_j))\\
 =  \frac{m_i}{2}+\frac{n_i-1}{2}+\frac{1}{2}\left(\sum_{j=1}^{\frac{n_i+1}{2}}\varepsilon^\alpha(x_j) - \sum_{j=\frac{n_i+3}{2}}^{n_i}\varepsilon^\alpha(x_j)\right) \enspace .
  \end{multline*}
  Observe that this is at least $\frac{m_i}{2}+\frac{n_i-1}{4}+\frac{1}{2}\varepsilon^\alpha(x_1) \ge \frac{m_i}{2}+\frac{n_i-1}{4} + \frac{r}{2}$.
  A similar argument applies when $\varepsilon^\alpha(x_j)=-r'<-r$.
  Therefore, if there exists $x \in A_i$ with an odd number of neighbors in $S$, then $|\varepsilon^\alpha(x)|\ge 1$ and so $\chi_\alpha(E_i) \ge \frac{m_i}{2}+\frac{n_i}{4} + \frac{1}{4}$.
  This is true whatever $\alpha$ assigns to $S$, and therefore $\mathbb{E}_\alpha[\chi_\alpha(E_i)] \ge \frac{m_i}{2}+\frac{n_i}{4} + \frac{1}{4}$.
  If there exists $x \in A_i$ with a non-zero even number of neighbors in $S$, then observe that either $\varepsilon^\alpha(x)=0$ or $|\varepsilon^\alpha(x)|\ge 2$.
  Furthermore, the probability that $\varepsilon^\alpha(x)=0$ is at most $\frac{1}{2}$, since given an assignment to all but one of the neighbors in $S$ of $x$, at most one of the possible assignments to the remaining neighbor will lead to $\varepsilon^\alpha(x)$ being $0$.
  Therefore, $\mathbb{E}_\alpha[\chi_\alpha(E_i)] \ge \frac{m_i}{2}+\frac{n_i-1}{4} + \frac{1}{2}\cdot 0 + \frac{1}{2} \cdot \frac{2}{2} = \frac{m_i}{2}+\frac{n_i}{4} + \frac{1}{4}$.
  We know that one of the above two cases must hold, since otherwise no vertex in $A_i$ has any neighbors in $S$, and Rule \ref{rule:kerRule1} applies.
  Therefore, if $|A_i|$ is odd, $\mathbb{E}_\alpha[\chi_\alpha(E_i)] \ge \frac{m_i}{2}+\frac{n_i}{4} + \frac{1}{4}$.

  Second, suppose that $|A_i|$ is even.
  Observe that if there exist $x,y \in A_i$ with $\varepsilon^\alpha(x)>\varepsilon^\alpha(y)$, then $\chi_\alpha(E_i) \ge \frac{m_i}{2}+\frac{n_i}{4}+\frac{1}{2}$.
  Indeed, then $\sum_{j=1}^{\frac{n_i}{2}}\varepsilon^\alpha(x_j) - \sum_{j=\frac{n_i}{2}+1}^{n_i}\varepsilon^\alpha(x_j) \ge 1$.
  So color $x_j$ with $\mathsf{red}$ if $j \le \frac{n_i}{2}$ and $\mathsf{blue}$ otherwise.
  Then the total number of satisfied edges in $E_i$ is equal to
  $$\frac{n_i}{2}\cdot\frac{n_i}{2} + \sum_{j=1}^{\frac{n_i}{2}}\varepsilon^\alpha_B(x_j) + \sum_{j=\frac{n_i}{2}+1}^{n_i}\varepsilon^\alpha_R(x_j) = \frac{m_i}{2}+\frac{n_i}{4} + \frac{1}{2}\left(\sum_{j=1}^{\frac{n_i}{2}}\varepsilon^\alpha(x_j) - \sum_{j=\frac{n_i}{2}+1}^{n_i}\varepsilon^\alpha(x_j)\right) \ge \frac{m_i}{2}+\frac{n_i}{4}+\frac{1}{2} \enspace .$$

  Also observe that if $|\varepsilon^\alpha(x)|\ge 2$ for all $x \in A_i$, then  $\chi_\alpha(E_i) \ge \frac{m_i}{2}+\frac{n_i}{4}+1$.
  Indeed, suppose $\varepsilon^\alpha(x)=r\ge 2$ for all $x \in A_i$. Then color $x_j$ with $\mathsf{red}$ if $j \le \frac{n_i}{2}+1$ and $\mathsf{blue}$ otherwise.
  Then the total number of satisfied edges in $E_i$ is equal to
  \begin{eqnarray*}
  \left(\frac{n_i}{2}+1\right)\left(\frac{n_i}{2}-1\right) + \sum_{j=1}^{\frac{n_i}{2}+1}\varepsilon^\alpha_B(x_j) + \sum_{j=\frac{n_i}{2}+2}^{n_i}\varepsilon^\alpha_R(x_j)
  & = & \frac{m_i}{2}+\frac{n_i}{4} - 1 + \frac{1}{2}\left(\sum_{j=1}^{\frac{n_i}{2}+1}\varepsilon^\alpha(x_j) - \sum_{j=\frac{n_i}{2}+2}^{n_i}\varepsilon^\alpha(x_j)\right)\\
  & = & \frac{m_i}{2}+\frac{n_i}{4} - 1 + \frac{1}{2}2r\\
  & \ge & \frac{m_i}{2}+\frac{n_i}{4} + 1 \enspace .
  \end{eqnarray*}
  A similar argument applies when $\varepsilon^\alpha(x)=-r\le-2$ for all $x \in A_i$.
  Finally observe that if $\varepsilon^\alpha(x)=0$ or $1$ for all $x \in A_i$, then by coloring $x_j$ with $\mathsf{red}$ if $j \le \frac{n_i}{2}$ and $\mathsf{blue}$ otherwise, $\chi_\alpha(E_i) \ge \frac{m_i}{2}+\frac{n_i}{4}$.
  Therefore, if there exist $x,y \in A_i$ such that $x$ has a neighbor in $S$ which is not adjacent to $y$, then the probability that $\varepsilon^\alpha(x)=\varepsilon^\alpha(y)$ is at most $\frac{1}{2}$ and so $\mathbb{E}_\alpha[\chi_\alpha(E_i)] \ge \frac{m_i}{2}+\frac{n_i}{4} + \frac{1}{2}\cdot 0 + \frac{1}{2} \cdot \frac{1}{2} = \frac{m_i}{2}+\frac{n_i}{4} + \frac{1}{4}$.

  Otherwise, all vertices in $A_i$ have the same neighbors in $S$.
  There must be at least two vertices in $S$ which are adjacent to the vertices in $A_i$, as otherwise Rule \ref{rule:kerRule1} or \ref{rule:kerRule2} would apply.
  Then the probability that $\varepsilon^\alpha(x)=0$ for every $x \in A_i$ is at most $\frac{1}{2}$, and the probability that $|\varepsilon^\alpha(x)|\ge2$ is at least $\frac{1}{4}$.
  (Consider any assignment to all but two of the neighbors of $x$ in $S$, and observe that of the four possible assignments to the remaining two vertices, at most two will lead to $\varepsilon^\alpha(x)=0$, and at least one will lead to $|\varepsilon^\alpha(x)|\ge 2$.)
  Therefore, $\mathbb{E}_\alpha[\chi_\alpha(E_i)] \ge \frac{m_i}{2}+\frac{n_i}{4} + \frac{1}{2}\cdot 0 + \frac{1}{4} \cdot 1 = \frac{m_i}{2}+\frac{n_i}{4} + \frac{1}{4}$.

  By linearity of expectation and the fact that $\chi_\alpha(E_L) = \sum_{i \in L}\chi_\alpha(E_i)$, it holds
  $$\mathbb{E}_\alpha [\chi_\alpha(E_L)] \ge \sum_{i \in L}\mathbb{E}[\chi_\alpha(E_i)] \ge  \sum_{i \in L}\left(\frac{m_i}{2}+\frac{n_i}{4}+\frac{1}{4}\right) = \frac{|E_L|}{2}+\frac{|V_L|}{4}+\frac{|L|}{4} \enspace .$$
  Therefore, there exists a partial assignment $\alpha$ on $S$ such that $\chi_\alpha(E_L) \ge \frac{|E_L|}{2} + \frac{|V_L|}{4} + \frac{|L|}{4}$, as required.  

  Now let $t'$ be the number of components in $G  -  (S \cup V_L)$, and observe that $t' \le \frac{|L|}{2}$. Indeed, if ${\cal C}$ is a component in $G  -  S$ then ${\cal C} -  V_L$ is connected if not empty.
  Furthermore, every component in $G - S$ contains at least one leaf-block, and if a component in~$G - S$ contains only one leaf-block, then the whole component is a single leaf block, and so ${\cal C}  -  V_L$ is empty.
  Therefore, every component in $G  -  (S \cup V_L)$ has at least two leaf-blocks.

  Let ${\cal C}_1, \hdots, {\cal C}_{t'}$ be the components of $G  -  (S \cup V_L)$.
  For each $j \in \{1,\hdots,t'\}$ let $n_j'$ be the number of vertices in $C_j$ and $m_j'$ the number of edges incident with vertices in $C_j$ (including edges between $C_j$ and $S \cup V_L$).
  Now let $S \cup V_L$ be colored such that the number of satisfied edges in $E_L$ is at least $\frac{|E_L|}{2} + \frac{|V_L|}{4} + \frac{|L|}{4}$.
  Note that this might mean that none of the edges in $G[S]$ are satisfied, but that there are at most $|S|^2$ of these.
  Now for each $j \in \{1,\hdots,t'\}$ let ${\cal C}_j$ be colored so that at least half the edges in $G[{\cal C}_j]$ plus $\frac{n_j'-1}{4}$ are satisfied, which can be done as this is the Edwards-Erd\H{o}s bound. 
  By reversing all the colors in ${\cal C}_j$ if needed, we can ensure that at least half the edges between ${\cal C}_j$ and the rest of the graph are satisfied, and therefore we can ensure at least $\frac{m_j'}{2}+\frac{n_j'-1}{4}$ of the edges incident with ${\cal C}_j$ are satisfied.

  We now have a complete assignment of colors to vertices, which satisfies at least
  \begin{eqnarray}
  \label{eqn:1}
    \frac{|E_L|}{2} + \frac{|V_L|}{4} + \frac{|L|}{4} + \sum_{j=1}^{t'}\left(\frac{m_j'}{2}+\frac{n_j'-1}{4}\right)
    & = & \frac{|E_L|}{2} + \sum_{j=1}^{t'}\frac{m_j'}{2} + \frac{|V_L|}{4} +\sum_{j=1}^{t'}\frac{n_j'}{4} + \frac{|L|-t'}{4}\\
    & \ge & \frac{m-|S|^2}{2} + \frac{n-|S|}{4} + \frac{|L|}{8}
  \end{eqnarray}
  edges. 
  Since $|L| \ge 4|S|^2+2|S|+2k-2$, the right-hand side of \eqref{eqn:1} is is at least
  $$\frac{m}{2}+\frac{n}{4} - \frac{(4|S|^2+2|S|)}{8} + \frac{4|S|^2+2|S|+2k-2}{8}  = \frac{m}{2}+\frac{n-1}{4} + \frac{k}{4} \enspace .$$
\end{proof}

Recall that $n^*$ is the number of blocks in $G - S$.

\begin{lemma}
\label{thm:noCliques}
  If $n^* \ge 4|S|^2+2|S|+4|L|+2k-7$ then $(G,k)$ is a ``yes''-instance.
\end{lemma}
\begin{proof}
  We will first produce a partial assignment on $G  -  S$ which satisfies $\frac{|E(G -  S)|}{2} + \frac{|V(G -  S)|-t}{4}$ edges, where $t$ is the number of connected components of $G-S$. Together with this we produce a set of vertices $R \subseteq V  \setminus  S$ such that we can change the color of any vertex in $R$ without changing the number of satisfied edges in $G -  S$.

  Observe that by choosing the right order on the blocks of $G - S$, it is possible to color all the blocks in such a way that each time we come to color a block, it contains at most one vertex that has already been colored.
  So color the vertices of each block $C_i$ such that if $|C_i|$ is even then half the vertices of $C_i$ are colored $\mathsf{red}$ and half are colored $\mathsf{blue}$, and if~$|C_i|$ is odd then $\frac{|C_i|+1}{2}$ of the vertices are colored $\mathsf{red}$ and $\frac{|C_i|-1}{2}$ are colored $\mathsf{blue}$.
  Furthermore, if~$|C_i|$ is odd and $A_i$ contains a vertex which is adjacent to $S$, then we ensure that at least one such vertex~$x$ is colored $\mathsf{red}$, and we add $x$ to $R$.

  Observe that the number of satisfied edges in a component ${\cal C}_j$ of $G  -  S$ is $\frac{|E(G[{\cal C}_j])|}{2} + \frac{|{\cal C}_j|+e_j-1}{4}$, where $e_j$ is the number of even sized blocks in ${\cal C}_j$.
  Let $e$ be the number of blocks in $G - S$ with an even number of vertices.
  Therefore, the total number of satisfied edges within $G  -  S$ is $\frac{|E(G -  S)|}{2} + \frac{|V(G  -  S)|+e-t}{4}$.
  Note also that $t\le |L|$.
  Observe that if we change the color of any vertex in $R$, the number of satisfied edges within ${\cal C}_j$ is unchanged. 

  First, suppose that $|R| \ge \frac{2|S|^2+|S|+|L|-e+k-1}{2}$; we will show that we have a ``yes''-instance.
  Color all the vertices in $S$ $\mathsf{red}$, or all $\mathsf{blue}$, whichever satisfies the most edges in $E(S, G -  S)$.
  If this satisfies at least $\frac{|E(S, G -  S)|}{2} + \frac{2|S|^2+|S|+|L|-e+k-1}{4}$ edges in $E(S, G -  S)$, then we are done, as the total number of satisfied edges is at least
  \begin{multline*}
    \frac{|E(G -  S)|}{2} + \frac{|E(S, G -  S)|}{2} + \frac{|V(G  -  S)|+e-t}{4} + \frac{2|S|^2+|S|+|L|-e+k-1}{4}\\
    \ge \frac{m-|S|^2}{2} + \frac{n-|S|}{4} + \frac{e-|L|}{4} + \frac{2|S|^2+|S|+|L|-e+k-1}{4} \ge \frac{m}{2} + \frac{n-1}{4} + \frac{k}{4} \enspace .
  \end{multline*}
  Otherwise, the number of satisfied edges in $E(S, G -  S)$ is between $\frac{|E(S, G -  S)|}{2}$ and $\frac{|E(S, G -  S)|}{2} +\linebreak \frac{2|S|^2+|S|+|L|-e+k-1}{4}$. Now change the colors of all the vertices in $R$ from $\mathsf{red}$ to $\mathsf{blue}$, and recall that this does not affect the number of satisfied edges within $G  -  S$.
  If the vertices of $S$ were all colored $\mathsf{red}$, then the number of satisfied edges in $E(S,G -  S)$ is increased by $|R|$ to at least  $\frac{|E(S,G -  S)|}{2} + \frac{2|S|^2+|S|+|L|-e+k-1}{2}$, and so again we are done.
  Otherwise, the vertices of $S$ were all colored $\mathsf{blue}$ and we lose $|R|$ satisfied edges, so the number of satisfied edges in $E(S,G -  S)$ is at most  $\frac{|E(S,G -  S)|}{2} - \frac{(2|S|^2+|S|+|L|-e+k-1)}{4}$.
  But then by changing the color of $S$ from $\mathsf{blue}$ to $\mathsf{red}$, we satisfy at least  $\frac{|E(S,G -  S)|}{2} + \frac{2|S|^2+|S|+|L|-e+k-1}{4}$ edges in $E(S, G  -  S)$, and we are done.

  So now we may assume that $|R| < \frac{2|S|^2+|S|+|L|-e+k-1}{2}$. Let $M \subseteq[n^*]$ be the set of all $i$ for which  either $|C_i|$ is odd and $A_i$ contains a vertex adjacent to $S$, or $|C_i|$ is even. Observe that $|M|=|R|+e < 2|R| + e < 2|S|^2+|S|+|L|+k-1$. Observe furthermore that $L \subseteq M$, as if there is a leaf-block $C_i$ with $|C_i|$ odd and no vertex in $A_i$ adjacent to $S$, we have an application of Rule \ref{rule:kerRule1}.
  It remains to bound $|[n^*]\setminus M|$. 
  Recall that $I \subseteq [n^*]$ is the set of all $i$ for which $|C_i \cap J|\ge 3$, and that $|L| \ge |I|+2$. 
  Therefore we have that $|M \cup I| < 2|S|^2+|S|+|L| + |L|-2+k-1 = 2|S|^2+|S|+2|L|+k-3$.
  Finally, for any $i \notin M \cup I$, it must be the case that $|C_i|$ is odd, $|C_i \cap J|=2$, and $A_i$ has no vertices adjacent to $S$. But then if $C_i, C_j$ share a vertex for any $i,j \notin M \cup I$, we have an application of Rule \ref{rule:kerRule3}. Therefore for $i \notin M \cup I$, $C_i$ provides an edge between $C_j$ and $C_{j'}$, for some $j,j' \in M \cup I$.
  From the fact that there are no cycles not contained in blocks, it can be seen that $|[n^*]\setminus (M\cup I)| \le |M \cup I|-1$.
  Thus, we have that $n^* \le 2|M \cup I| - 1 < 4|S|^2+2|S|+4|L|+2k-7$.

\end{proof}


\begin{lemma}
\label{thm:cliqueSize}
  If $|A_i| \ge   2|S|^3 + 5|S|^2 + (|L|+k-3)|S| -2|L|-|J|-2k$ for some $i\in \{1,\hdots,n^*\}$ then $(G,k)$ is a ``yes''-instance.
\end{lemma}
\begin{proof}
  Let $i$ be such that $|A_i| \ge 2|S|^3 + 5|S|^2 + (|L|+k-3)|S| -2|L|-|J|-2k$.
  Let $n' = |A_i|$ and let $m'$ be the number of edges within $A_i$ and between $A_i$ and $S$.
  Observe first that if we can find an assignment to $S \cup A_i$ that satisfies at least $\frac{m'}{2}+\frac{n'-1}{4} + \frac{2|S|^2+|S|+|L|+k}{4}$ of these edges, then we have a ``yes''-instance: indeed, the component of $G  -  S$ containing $A_i$ is still connected in $G -  (S \cup A_i)$ unless it is empty.
  Therefore, the number of components in~$G -  (S \cup A_i)$ is no more than in $G  -  S$ and is therefore at most $t \le |L|$. Let ${\cal C}_1, {\cal C}_2, \hdots, {\cal C}_t$ be the components of $G -  (S \cup A_i)$.
  Color each component ${\cal C}_j$ optimally so that at least $\frac{E(G[{\cal C}_j])}{2} + \frac{|{\cal C}_j|-1}{4}$ of the edges within ${\cal C}_j$ are satisfied, and then reverse the colors of ${\cal C}_j$ if necessary to ensure that at least half the edges between ${\cal C}_j$ and $S \cup A_i$ are satisfied.
  Then the total number of satisfied edges is at least
  \begin{multline*}
  \frac{m-|E(G[S \cup  A_i])|}{2} + \frac{n-|S \cup  A_i|-t}{4} + \frac{m'}{2}+\frac{n'-1}{4} + \frac{2|S|^2+|S|+|L|+k}{4}\\
\ge \frac{m-m'-|S|^2}{2} + \frac{n-n'-|S|-|L|}{4} + \frac{m'+|S|^2}{2}+\frac{n'+|S|+|L|+k-1}{4} = \frac{m}{2}+\frac{n-1}{4}+\frac{k}{4},
   \end{multline*}
  and so we have a ``yes''-instance.

  We now show that if $n'\ge   2|S|^3 + 5|S|^2 + (|L|+k-3)|S| -2|L|-|J|-2k$ then we can find a partial assignment to $S \cup A_i$ that satisfies at least $\frac{m'}{2}+\frac{n'-1}{4} + \frac{2|S|^2+|S|+|L|+k}{4}$ of the edges within~$G[A_i]$ and between $A_i$ and $S$, completing the proof.
  For a given partial assignment $\alpha$ on $S$, and for any vertex $x \in A_i$, let $\varepsilon^\alpha_R(x)$ be the number of neighbors of $x$ which are assigned $\mathsf{red}$, let $\varepsilon^\alpha_B(x)$ be the number of neighbors of $x$ which are assigned $\mathsf{blue}$, and let $\varepsilon^\alpha(x) = \varepsilon^\alpha_B(x)-\varepsilon^\alpha_R(x)$.
  Let $x_1, \hdots, x_l$ be the vertices of $A_i$, ordered so that $\varepsilon^\alpha(x_1) \ge \varepsilon^\alpha(x_2) \ge \hdots \ge \varepsilon^\alpha(x_l)$ (the ordering will depend on~$\alpha$).
  Suppose for some $r \in \{1,\hdots,n^*\}$ we assign $x_j$ $\mathsf{red}$ for $j \le \frac{n'+r}{2}$, and the assign the remaining $\frac{n'-r}{2}$ vertices in $A_i$ $\mathsf{blue}$.
  Then the number of satisfied edges is $\frac{m'}{2} + \frac{n'}{4} - \frac{r^2}{4}+\frac{1}{2}\left(\sum_{j=1}^{\frac{n'+r}{2}}\varepsilon^\alpha(x_j) - \sum_{j=\frac{n'+r+2}{2}}^{n'}\varepsilon^\alpha(x_j)\right)$.
  Suppose there exist $X_1, X_2 \subseteq A_i$ with $|X_1|, |X_2| \ge \frac{2|S|^2+|S|+|L|+k}{2}$ and such that $\varepsilon^\alpha(x)>\varepsilon^\alpha(y)$ for all $x\in X_1, y\in X_2$.
  Observe that by setting $r=0$ (if $n'$ is even) or $r\in\{\pm 1\}$ (if $n'$ is odd), we get that
  \begin{equation*}
  \frac{1}{2}\left(\sum_{j=1}^{\frac{n'+r}{2}}\varepsilon^\alpha(x_j) - \sum_{j=\frac{n'+r+2}{2}}^{n'}\varepsilon^\alpha(x_j)\right) \geq \frac{2|S|^2+|S|+|L|+k}{4},
  \end{equation*}
  and so we are done.

  It remains to show that there is some partial assignment $\alpha$ on $S$ such that $X_1,X_2$ exist.
  For the sake of contradiction, suppose this were not the case.
  Then for any assignment $\alpha$ there exist a set of at least $n'-2|S|^2-|S|-|L|-k$ vertices $x$ in $A_i$ for which $\varepsilon^\alpha(x)$ is the same.
  Consider first the partial assignment $\alpha$ for which every vertex in $S$ is assigned $\mathsf{red}$.
  Then for every $x \in A_i$, $\varepsilon^\alpha(x) = |N(x)\cap S|$.
  It follows that $|N(x)\cap S|$ is the same for at least $n'-2|S|^2-|S|-|L|-k$ vertices~$x$ in~$A_i$.
  Let $A_i'$ be this set of vertices.
  Now for each vertex $s \in S$, consider the partial assignment $\alpha$ for which $s$ is assigned $\mathsf{blue}$ and every other vertex in $S$ is assigned $\mathsf{red}$. Let $X_1$ be the set of vertices in $A_i'$ which are adjacent to $s$, and let $X_2$ be the set of vertices in $A_i'$ which are not adjacent to $s$.
  Then $\varepsilon^\alpha(x) = \varepsilon^\alpha(y)$ if $x,y \in X_1$ or $x,y \in X_2$, and  $\varepsilon^\alpha(x) = \varepsilon^\alpha(y)-2$ if $x \in X_1, y \in X_2$.
  Therefore by our assumption either $|X_1| < \frac{2|S|^2+|S|+|L|+k}{2}$ or  $|X_2| < \frac{2|S|^2+|S|+|L|+k}{2}$.
  It follows that for each $s \in S$ there are either at most $\frac{2|S|^2+|S|+|L|+k}{2}$ vertices in $A_i'$ which are adjacent to $s$, or at most $\frac{2|S|^2+|S|+|L|+k}{2}$ vertices in $A_i'$ which are not adjacent to $s$.
  Therefore, there is a set of at least
  $$|A_i'| - |S|\frac{(2|S|^2+|S|+|L|+k)}{2} \ge n'-2|S|^2-|S|-|L|-k - |S|\frac{2|S|^2+|S|+|L|+k}{2}$$ vertices in $A_i$ which are all adjacent to exactly the same vertices in $S$.
  If
  $$
  2|S|^2-|S|-|L|-k + |S|\frac{2|S|^2+|S|+|L|+k}{2}\ge \frac{|A_i|+|J|+|S|}{2}$$
  then we would have an application of Rule \ref{rule:kerRule4}.
  Therefore,
  \begin{eqnarray*}
  |A_i| & \le &  4|S|^2-2|S|-2|L|-2k  + |S|(2|S|^2+|S|+|L|+k) - |J|-|S|\\
        & = & 2|S|^3 + 5|S|^2 + (|L|+k-3)|S| -2|L|-|J|-2k \enspace .
  \end{eqnarray*}
%
%
%
\end{proof}

To complete the proof of Theorem \ref{thm:polykernelstrongerbound}, we prove Lemma \ref{thm:noVertices}, from which Theorem \ref{thm:polykernelstrongerbound} follows.

\begin{lemma}
\label{thm:noVertices}
  For a connected graph $G$ that is reduced by Rules \ref{rule:kerRule1}, \ref{rule:kerRule2}, \ref{rule:kerRule3} and \ref{rule:kerRule4} and satisfies 
  $|V(G)| >  29160k^5 + 6480k^4 -8532k^3 - 492k^2 + 731k - 80$, the pair $(G,k)$ is a ``yes''-instance.
\end{lemma}
\begin{proof}
  Observe that $|J|\le n^*-1$.
  By the preceding three lemmas, we may assume that
  \begin{enumerate}
    \item $|L|< 4|S|^2+2|S|+2k-2$, and
    \item $n^* < 4|S|^2+2|S|+4|L|+2k-7$, and
    \item $|A_i| <   2|S|^3 + 5|S|^2 + (|L|+k-3)|S| -2|L|-|J|-2k$ for any $i \in \{1,\hdots,n^*\}$,
  \end{enumerate}
  for otherwise we have a ``yes''-instance.
  Furthermore, we know that $|S|<3k$.
  Putting everything together, we have that
  $$|L|< 4|S|^2+2|S|+2k-2 < 36k^2+8k-2 \enspace .$$
  This in turn implies that
  $$n^* < 4|L|+4|S|^2+2|S|+2k-7 <  180k^2+40k-15,$$ and so $|J| < 180k^2+40k-16$.
  For every $i \in \{1,\hdots,n^*\}$, we have
  \begin{eqnarray*}
  |A_i| & < &   2|S|^3 + 5|S|^2 + (|S|-2)|L| + (k-3)|S| -|J|-2k\\
        & < & 54k^3+45k^2+(3k-2)(36k^2+8k-2) + 3k^2-9k-|J|-2k\\
        & \le & 162k^3 - 33k + 4 \enspace .
  \end{eqnarray*}
  (Here we assume that $|S|\ge 2$ and and $k \ge 3$, as otherwise the problem can be solved in polynomial time.)
  Observe that $V = \bigcup_{i\in n^*} A_i \cup J \cup S$.
  Therefore, the number of vertices in $G$ is at most
  \begin{multline*}
  (180k^2+40k-16)(162k^3 - 33k + 4)+ (180k^2+40k-16) + 3k\\
 = 29160k^5 + 6480k^4 -8532k^3 - 600k^2 + 688k - 64 + 108k^2 + 43k -16\\
 = 29160k^5 + 6480k^4 -8532k^3 - 492k^2 + 731k - 80 \enspace .
 \end{multline*}
\end{proof}

\section{Discussion and Open Problems}
\label{sec:discussion}
We showed fixed-parameter tractability of {\sc Max-Cut} parameterized above the Edwards-Erd\H{o}s bound $m/2 + (n-1)/4$, and thereby resolved an open question from~\cite{MahajanRaman1997,MahajanEtAl2009,GutinYeo2010,Sikdar2010,CrowstonEtAl2011}.
Furthermore, we showed that the problem has a kernel with $O(k^5)$ vertices.
We have not attempted to optimize running time or kernel size, and indeed we conjecture that {\sc Max-Cut} has a kernel with $O(k^3)$ vertices and the edge version admits a linear kernel.
Our conjecture was recently answered in the affirmative, by Crowston et al.~\cite{CrowstonEtAl2013}.

It remains an open problem whether the weighted version of {\sc Max-Cut} above the Edwards-Erd\H{o}s bound is fixed-parameter tractable; our conjecture is that this problem is also fixed-parameter tractable and admits a polynomial kernel.

Since an extended abstract of this paper appeared~\cite{CrowstonEtAl2012}, it was shown~\cite{CrowstonEtAl2012b,MnichEtAl2012} that the techniques developed in this paper are strong enough to show fixed-parameter tractability of the {\sc Maximum Acyclic Subdigraph} problem in oriented graphs parameterized above tight lower bound; this solves another open question by Raman and Saurabh~\cite{RamanSaurabh2006} in the field of parameterized complexity.

\paragraph{Acknowledgment.} We thank Tobias Friedrich and Gregory Gutin for help with the presentation of the results. Part of this research has been supported by an International Joint Grant from the Royal Society.



\end{document}